\documentclass[12pt,a4paper]{article}

\usepackage{amsmath,amsthm,amssymb,amscd,a4wide}

\usepackage{dsfont}
\usepackage{mathrsfs}
\usepackage{setspace}
\usepackage{hyperref}

\newcommand{\ud}{\mathrm{d}}

\newcommand{\cD}{{\mathcal D}}

\numberwithin{equation}{section}

\newtheorem{theorem}{Theorem}[section]
\newtheorem{lemma}[theorem]{Lemma}
\newtheorem{prop}[theorem]{Proposition}

\newtheorem{remark}[theorem]{Remark}

\theoremstyle{definition}

\numberwithin{equation}{section}

\begin{document}

\thispagestyle{empty}

\vspace*{1cm}

\begin{center}

{\LARGE\bf On bound electron pairs in a quantum wire } \\

\vspace*{2cm}

{\large Joachim Kerner \footnote{E-mail address: {\tt Joachim.Kerner@fernuni-hagen.de}} }%

\vspace*{5mm}

Department of Mathematics and Computer Science\\
FernUniversit\"{a}t in Hagen\\
58084 Hagen\\
Germany\\

\end{center}

\vfill

\begin{abstract} Based on the quantum two-body problem introduced in \cite{KM16Rep} we consider bound pairs of electrons moving on the positive half-line. The analysis is motivated by the ground-breaking work of Cooper \cite{CooperBoundElectron} who identified the pairing of electrons as a possible explanation for superconducting behaviour in metals. In this paper we are interested in the connection between the topologies of the underlying one- and two-particle configuration spaces and spectral properties of the Hamiltonian linked to a condensation of pairs. We derive explicit estimates for the energy gap and prove condensation for a gas of non-interacting pairs. Finally, we add some disorder to the system and prove destruction of the condensate.
\end{abstract}

\newpage

\section{Introduction}
This paper is motivated by the work of Cooper \cite{CooperBoundElectron} and its application shortly after in the explanation of (type-I) superconducting behaviour in metals, a joint work with Bardeen and Schrieffer \cite{BCSI}. 

Superconductivity refers to a vanishing electrical resistance of metals at low temperatures and was experimentally observed by K.~Onnes in 1911 \cite{Onnes1991}. This sudden change of resistance below some critical temperature was afterwards interpreted as a phase transition in a many-particle system similar to the one occurring in free Bose gases, namely, Bose-Einstein condensation (BEC) \cite{EinsteinBEC,MR04}. BEC, on the other hand, refers to the macroscopic occupation of a single-particle state in the so-called thermodynamic limit for temperatures lower than a critical one \cite{PO56}. In this way the wave function becomes a macroscopic quantity and an order parameter. However, since electrons are fermions which obey the Pauli exclusion principle, a single state can be occupied by at most two electrons and it hence remained unclear how one could arrive at such a phase transition to explain superconductivity. The solution, as suggested by Bardeen, Cooper and Schrieffer, lies in a pairing of electrons (Cooper pairs) and the condensation then occurs not for single electrons but for pairs of them. Most importantly, since the pairs are formed due to the existence of a net attractive interaction between the electrons (electron-phonon-electron interaction), the energy of the BCS-ground state is lower than the one of the non-interacting ground state. Furthermore, there exists a finite (volume-independent)  energy gap $\Delta > 0$ that separates the BCS-ground state from the first excited state which itself is obtained by breaking up a Cooper pair in the ground state. At this point it should be mentioned that the existence of an energy gap in the superconducting regime was arguably the most important experimentally confirmed fact in this context. The derivation of such a gap through a pairing of electrons was then the main objective in Cooper's paper \cite{CooperBoundElectron}.

In this paper we are also interested in investigating the effect of a pairing of electrons based on the quantum two-body Hamiltonian constructed in \cite{KM16Rep}. In this model the two electrons are restricted to have no distance larger than some given value $d > 0$ and hence the pair cannot be broken by adding energy to the system. However, the Hamiltonian exhibits interesting spectral properties which are directly linked to the topologies of the one- and two-particle configuration spaces. Most importantly, the topology of the two-particle configuration space leads to the existence of a spectral gap which subsequently is responsible for the condensation of pairs. Quite interestingly, we are also able to relate $d > 0$ to the excitation energy of the pair and consequently to a energy gap resembling $\Delta$. We then prove Bose-Einstein condensation in a gas of non-interacting pairs and finally, by adding some disorder to the system in form of a repulsion between the electrons, we prove the destruction of the condensate.

\section{The Hamiltonian for a single electron pair}
We consider two electrons on the half-line $\mathbb{R}_+=[0,\infty)$. Without loss of generality we also assume that both electrons have the same spin which eventually implies that the wave function has to be anti-symmetric. As in \cite{KM16Rep}, the pair shall be described by the Hamiltonian 
\begin{equation}\label{HamiltonianCooperPair}
H_{p}=-\frac{\hbar^2}{2m_e}\left(\frac{\partial^2}{\partial x^2}+\frac{\partial^2}{\partial y^2}\right)+v_{bin}(|x-y|)\ ,
\end{equation}
with binding potential $v_{bin}:\mathbb{R}_+ \rightarrow [0,\infty]$,
\begin{equation}\label{BindingPotential}
v_{bin}(x):= \begin{cases} 0 \quad \text{if}\quad 0\leq x \leq d\ , \\
\infty \quad \text{else}\ .
\end{cases}
\end{equation}
In this way, the parameter $d > 0$ describes the spatial extension of the pair. Furthermore, the chosen binding potential effectively leads to a reduction of the two-particle configuration space from $\mathbb{R}^2_+$ to the set (pencil-shaped domain)
\begin{equation}\label{PencilDomain}
\Omega:=\{(x,y) \in \mathbb{R}^2_+\ |\  |x-y|\leq d  \}\ ,
\end{equation}
on which $H_p$ then acts as the standard two-dimensional Laplacian. Also, the quadratic form associated with $H_{p}$ is defined on the Hilbert space 
\begin{equation}
L^2_a(\Omega):=\{\varphi \in L^2(\Omega)\ |\  \varphi(x,y)=-\varphi(y,x)\}
\end{equation}
and it is given by 
\begin{equation}
q[\varphi]=\frac{\hbar^2}{2m_e}\int_{\Omega}|\nabla \varphi|^2\ \ud x\,
\end{equation}
with domain $\cD_q=\{\varphi \in H^1(\Omega)\ |\  \varphi(x,y)=-\varphi(y,x) \ \text{and}\ \varphi|_{\partial \Omega_D}=0  \}$. 

Note that $\partial \Omega_D:=\{(x,y)\in \Omega\ |\ |x-y|=d  \}$. For later purposes we also introduce $\partial \Omega_{\sigma}:=\{(x,y)\in \Omega\ |\ x=0 \ \text{or}\ y=0 \}$. Also note that the anti-symmetry effectively leads to extra Dirichlet boundary conditions along the line $y=x$ for functions $\varphi \in \cD_q$.
\subsection{Spectral properties of $H_{p}$}
In this subsection we want to characterise the spectral properties of the pair Hamiltonian $H_p$. In particular, we want to prove the existence of an energy gap $\Delta=\Delta(d) > 0$ that separates the ground state from the first excited state. 
\begin{theorem}\label{TheoremEssentialSpectrum} For the essential spectrum one has 
\begin{equation}
\sigma_{ess}(H_{p})=\left[\frac{\hbar^2 \pi^2}{m_ed^2},\infty\right)\ .
\end{equation}
Furthermore, regarding the discrete part of the spectrum one has  
\begin{equation}
\sigma_{d}(H_{p})\neq \emptyset\ .
\end{equation}
\end{theorem}
\begin{proof} The first part of the statement follows adapting the proof of [Theorem~2,\cite{KM16Rep}] with only one small difference: Namely, due to anti-symmetry one has additional Dirichlet boundary conditions along the axis $y=x$. Hence, the appropriate Weyl sequence again consists of rectangles but with half of the width. This then leads to an extra factor of four and hence to the given value for the bottom of the essential spectrum.
	
The second part of the statement, on the other hand, follows using similar methods as in the proof of [Theorem~3,\cite{KM16Rep}]: In a first step one restricts $\Omega$ to the domain $\tilde{\Omega}:=\{(x,y) \in \mathbb{R}^2_+\ |\  |x-y|\leq d\ , \ y \geq x  \}$. Then one reflects $\tilde{\Omega}$ across the axis $x=0$ to obtain the (l-shaped) domain $\hat{\Omega}$. Since one then has Dirichlet boundary conditions along all boundary segments of $\hat{\Omega}$, the statement follows directly from \cite{ExnerL}.
\end{proof}
\begin{remark} The existence of a finite spectral gap that separates the ground state from the essential spectrum is a highly non-trivial issue which is strongly related to the geometry and hence topology of the two-particle configuration space, i.e., the pencil-shaped domain $\Omega$ (see [Remark~4,\cite{KM16Rep}] for a further discussion).
\end{remark}
Now, it is actually possible to establish an estimate for the ground state energy $E_0=\inf \sigma(H_{p})$ of $H_{p}$. 
\begin{theorem}[Ground state energy of an electron pair]\label{GroundStateEnergyCP} Let $E_0$ denote the lowest eigenvalue of $H_{p}$. Then one has 
	\begin{equation}\label{EquationGroundStateEnergyCooperPair}
	0,25\cdot \frac{\hbar^2\pi^2}{m_ed^2} \leq E_0 \leq 0,93 \cdot \frac{\hbar^2 \pi^2}{m_ed^2}\ .
	\end{equation}
	\end{theorem}
	\begin{proof} We first note that the upper bound follows directly using the results of \cite{ExnerL}, see also the proof of [Theorem~3,\cite{KM16Rep}].
		
		In order to obtain the lower bound, consider any normalised function $\varphi \in \cD_q$ such that $q[\varphi] < \frac{\hbar^2 \pi^2}{m_ed^2}$. Then consider the restriction of $\varphi$ onto the triangle $\Omega_{triangle}:=\{(x,y)\in \Omega\ |\ 0 \leq x,y\leq d \}$. This restriction as well as its gradient cannot vanish identically (see the proof of [Theorem~3,\cite{KM16Rep}]). Then, due to the Dirichlet boundary conditions along the axis $y=x$ we obtain (note that we have a triangle with Dirichlet and Neumann boundary conditions and hence the ground state eigenvalue can be determined explicitly), also using inequality~(15) of \cite{KM16Rep},
		\begin{equation}
		\frac{\hbar^2\pi^2}{4m_ed^2} \leq \frac{\hbar^2}{2m_e}\frac{\int_{\Omega_{triangle}}|\nabla \varphi|^2\ \ud x}{\int_{\Omega_{triangle}}|\varphi|^2\ \ud x} \leq E_0\ .
		\end{equation}

	\end{proof}
	In the next result we prove that there actually exists exactly one bound state, i.e., one eigenstate below the bottom of the essential spectrum.
	\begin{theorem}\label{ExistenceOneEigenvalue} One has $\sigma_d(H_p)=\{E_0 \}$. Furthermore, the multiplicity of the ground state is one.
	\end{theorem}
	\begin{proof}
		We first note that uniqueness of the ground state (up to a phase) follows by standard arguments [Theorem~11.8,\cite{LiebLoss}].
		
		To prove that there exists no other eigenvalue below $\inf \sigma_{ess}(H_p)$  we use a operator-bracketing argument \cite{BEH08}: Consider the two-dimensional Laplacian $-\frac{\hbar^2}{2m_e}\Delta$ on the cross-shaped domain 
	\begin{equation}\begin{split}
	\Omega_{b}:=\{(x,y)\in \mathbb{R}^2|\ &x\in (-\infty,+\infty) \ \text{and}\ 0  \leq y \leq b  \} \\
	&\cup \{(x,y)\in \mathbb{R}^2|\ y\in (-\infty,+\infty) \ \text{and}\ 0  \leq x \leq b  \}\ ,
	\end{split}
	\end{equation}
	subjected to Dirichlet boundary conditions. Using similar arguments as in the proof of Theorem~\ref{TheoremEssentialSpectrum}, one readily sees that the essential spectrum of this Laplacian starts at $\frac{\hbar^2\pi^2}{2m_eb^2}$. Now consider the comparison operator $-\frac{\hbar^2}{2m_e}\Delta|_{\Omega_{s}} \oplus -\frac{\hbar^2}{2m_e}\Delta|_{\Omega_b\setminus \Omega_{s}}$ where $\Omega_{s}:=[0,b] \times [0,b]$. Note that $-\frac{\hbar^2}{2m_e}\Delta|_{\Omega_{s}}$ is subjected to Neumann boundary conditions and $-\frac{\hbar^2}{2m_e}\Delta_2$ on $\Omega_{b}\setminus\Omega_{s}$ fulfils Neumann boundary conditions along the boundary segments touching the square and Dirichlet boundary conditions elsewheren. Then
	\begin{equation}
	-\frac{\hbar^2}{2m_e}\Delta|_{\Omega_{s}} \oplus -\frac{\hbar^2}{2m_e}\Delta|_{\Omega_b\setminus \Omega_{s}} \leq -\frac{\hbar^2}{2m_e}\Delta|_{\Omega_{b}}
	\end{equation}
	in the sense of operators. Now, the important thing to note is that $-\frac{\hbar^2}{2m_e}\Delta|_{\Omega_{s}}$ has purely discrete spectrum with eigenvalues $\{0,\frac{\hbar^2\pi^2}{2m_eb^2},...\}$ and hence there exists only one eigenvalue smaller than $\frac{\hbar^2\pi^2}{2m_eb^2}$. Since the spectrum of $-\frac{\hbar^2}{2m_e}\Delta|_{\Omega_b\setminus \Omega_{s}}$ also starts at $\frac{\hbar^2\pi^2}{2m_eb^2}$, we conclude that $-\frac{\hbar^2}{2m_e}\Delta|_{\Omega_{b}}$ has at most one isolated eigenvalue (it has exactly one).
	
	To obtain the statement we consider the operator  $-\frac{\hbar^2}{2m_e}\Delta$ on the domain $\hat{\Omega}$ as defined in the proof of Theorem~\ref{TheoremEssentialSpectrum}. One has 
		\begin{equation}
		-\frac{\hbar^2}{2m_e}\Delta|_{\Omega_b}\leq -\frac{\hbar^2}{2m_e}\Delta |_{\hat{\Omega}} \oplus 0|_{\Omega_b\setminus \hat{\Omega}}\ ,
		\end{equation}
	identifying $b:=d/\sqrt{2}$. Finally, by the same reasoning as above we conclude that $-\frac{\hbar^2}{2m_e}\Delta |_{\hat{\Omega}} \oplus 0|_{\Omega_b\setminus \hat{\Omega}}$ has at most one isolated eigenvalue below $\frac{\hbar^2\pi^2}{m_ed^2}$ and so has our original operator by the reflection symmetry of the domain $\hat{\Omega}$.
	\end{proof}

Based on the previous theorems one now has a clear picture of the excitation spectrum of the pair of electrons in the considered model. As already mentioned in the introduction, it is not possible to break up the pair but it is nevertheless possible to excite the pair by adding kinetic energy to it. Hence, one obtains for the energy gap the relation
\begin{equation}\label{EnergyGap}
\Delta=\Delta(d) \sim \frac{\hbar^2 \pi^2}{m_e d^2}\ ,
\end{equation}
when defining $\Delta(d):=\frac{\hbar^2 \pi^2}{m_ed^2}-E_0$. According to \cite{MR04} the energy gap in superconducting metals is of order $10^{-3}$eV. Plugging this into \eqref{EnergyGap} one concludes that $d$ has to be of order $10^{-6}$m which is in agreement with the value for the extension of Cooper pairs as derived in \cite{CooperBoundElectron}.

\section{On the condensation of bound electron pairs}
\label{Sec2}
As already mentioned in the introduction, it is the condensation of the (Cooper) pairs into a single state which constitutes the superconducting phase with its characteristic coherent many-particle behaviour. It is therefore the aim of this section to discuss Bose-Einstein condensation (BEC) for a gas of non-interacting electron pairs, each described by the Hamiltonian \eqref{HamiltonianCooperPair}.

As customary in statistical mechanics, BEC will be investigated in the grand-canonical ensemble using a suitable thermodynamic limit~\cite{RuelleSM}. For this one first reduces the one-pair configuration space to $\Lambda_{L}:=[0,L]$ and, accordingly, the Hilbert space of a pair is then given by $L^2(\Omega_L)$ with 
\begin{equation}\label{PencilDomainTL}
\Omega_L:=\{(x,y) \in \mathbb{R}^2_+\ |\  |x-y|\leq d \ \text{and}\ 0\leq x,y\leq L  \}\ .
\end{equation}
To arrive at a self-adjoint realisation of the pair Hamiltonian~\eqref{HamiltonianCooperPair} on $L^2(\Omega_L)$ one has to introduce extra boundary conditions along the segments of $\Omega_L$ where $x=L$ or $y=L$. For convenience we choose Dirichlet boundary conditions and denote the corresponding self-adjoint version of~\eqref{HamiltonianCooperPair} as $H^{L}_p$. The operator $H^{L}_p$ has purely discrete spectrum and we write its eigenvalues as $\{E_L(n)\}_{n \in \mathbb{N}_0}$. Note that the quadratic form associated with $H^{L}_p$ is given by 
\begin{equation}
q_L[\varphi]=\frac{\hbar^2}{2m_e}\int_{\Omega_L}|\nabla \varphi|^2\ \ud x\,
\end{equation}
with domain $\cD_{q_L}=\{\varphi \in H^1(\Omega_L)\ |\  \varphi(x,y)=-\varphi(y,x)\ \text{and}\ \varphi|_{\partial \Omega^L_D}=0  \}$ where $\partial \Omega^L_D:=\{(x,y) \in \Omega_L|\ |x-y|=d \ \text{or}\ x=L \ \text{or}\ y=L \}$.

The thermodynamic limit is then realised as the limit $L \rightarrow \infty$ such that 
\begin{equation}\label{EquationThermodynamicLimit}
\frac{1}{L}\sum_{n=0}\frac{1}{e^{\beta\left(E_L(n)-\mu_L\right)}}=\rho
\end{equation}
holds for all values of $L >0$ with $\rho > 0$ being the density of pairs, $\beta=1/T$ the inverse temperature and $\mu_L \leq E_L(0)$ the sequence of chemical potentials~\cite{LandauWildeBEC,RuelleSM}. 
\begin{lemma}\label{LemmaConvergenceGroundstate} One has $E_L(0) \rightarrow E_0$ as $L \rightarrow \infty$.
\end{lemma}
\begin{proof} By the min-max principle we have $E_0 \leq E_L(0)$ since the eigenfunction corresponding to $E_L(0)$ can be extended by zero onto $\Omega$. On the other hand, using a suitable sequence of test functions $\tau_L \subset C^{\infty}(\Omega)$ one has that $\tau_L \varphi_0 \in \cD_{q_L}$ as well as $\tau_L \varphi_0 \rightarrow \varphi_0$ in $H^1(\Omega)$, $\varphi_0$ being the eigenfunction associated with $E_0$. Hence, one has the inequality $E_0 \leq E_L(0) \leq E_0 + \varepsilon_L$ with $\varepsilon_L \rightarrow 0$ as $L \rightarrow \infty$ thus proving the statement. 
\end{proof}
Using a bracketing argument similar to the one used in the proof of Theorem~\ref{ExistenceOneEigenvalue} one arrives at the following statement. 
\begin{prop} One has $\frac{\hbar^2 \pi^2}{m_ed^2} \leq E_L(n)$ for all $n\geq 1$ and all $L > d$.
	\end{prop}
Comparing with~\eqref{EquationThermodynamicLimit} we see that the density of pairs in the excited states with eigenvalues $\{E_L(n)\}_{n \geq 1}$ is given by 
\begin{equation}
\rho_{ex}(\beta,\mu_L,L)=\frac{1}{L}\sum_{n=1}^{\infty}\frac{1}{e^{\beta(E_L(n)-\mu_L)}-1}\ .
\end{equation}
Using a bracketing argument similar to the one used in the proof of Theorem~\ref{ExistenceOneEigenvalue} one can show that (setting the chemical potential constant)
\begin{equation}\label{ParticleDensityExcitedStates}
\rho^{ex}_{\infty}(\beta,\mu):=\lim_{L \rightarrow \infty}\rho_{ex}(\beta,\mu,L)=\frac{\sqrt{4m_e}}{\hbar \pi}\sum_{n=1}^{\infty}\int_{0}^{\infty}\frac{1}{e^{\beta\frac{\hbar^2\pi^2}{m_ed^2}n^2}e^{\beta (x^2-\mu)}-1}\ud x\ .
\end{equation}
Formally, \eqref{ParticleDensityExcitedStates} can be derived by taking into account that the (contributing) eigenvalues $\{E_L(n)\}_{n\geq 1}$ are approximately given by $\left\{E_L(k,l) \sim \frac{\hbar^2}{2m_e}\left(\frac{2\pi^2}{d^2}k^2+\frac{\pi^2}{8L^2}(2l+1)^2\right)\right\}$, where we have double-indexed them with $k,l \in \mathbb{N}$, since those are the eigenvalues of a rectangle $[0,\sqrt{2}L] \times [0,d/\sqrt{2}]$ with Neumann boundary conditions at the left segment and Dirichlet boundary conditions elsewhere (note here the structure of $\Omega_L$). In the limit $L \rightarrow \infty$, one of the sums is a Riemann sum consequently turning into an integral. 

Now, we say that a condensation of (electron) pairs occurs in the thermodynamic limit if there exists an eigenstate $\varphi_L(n)$ of $H^L_p$, being the one associated with the eigenvalue $E_L(n)$, which is macroscopically occupied in the thermodynamic limit, i.e., if
\begin{equation}
\limsup_{L \rightarrow \infty}\frac{1}{e^{\beta\left(E_L(n)-\mu_L\right)}-1} > 0 \  .
\end{equation}
\begin{theorem}[Condensation of electron pairs]\label{CondensationPairs} For given inverse temperature $\beta \in (0,\infty)$ there exists a critical pair density $\rho_{crit}(\beta) > 0$ such that there is a condensation of pairs for all densities $\rho > \rho_{crit}(\beta)$ in the thermodynamic limit.
\end{theorem}
\begin{proof} Since $\mu_L \leq E_L(0)$ and $E_L(0) \rightarrow E_0=\frac{\hbar^2 \pi^2}{m_ed^2}-\Delta(d)$ as $L \rightarrow \infty$ one concludes that, for any $\varepsilon_1 > 0$,  $\mu_L < E_0 +\varepsilon_1$ for all $L$ large enough. Accordingly we have that $\mu_L < \frac{\hbar^2\pi^2}{m_ed^2}-\varepsilon_2:=\mu$ for some small $\varepsilon_2 > 0$ and $L$ large enough. This implies
	\begin{equation}\label{EquationProofBECPairs}\begin{split}
	\frac{1}{L}\sum_{n=1}^{\infty}\frac{1}{e^{\beta\left(E_L(n)-\mu_L\right)}-1} &\leq \frac{1}{L}\sum_{n=1}^{\infty}\frac{1}{e^{\beta\left(E_L(n)-\mu\right)}-1}\\
	&\leq \rho_{ex}(\beta,\mu) + \varepsilon_3(L)\ ,
	\end{split}
	\end{equation}
where $\varepsilon_3(L) \rightarrow 0$ as $L \rightarrow \infty$ by~\eqref{ParticleDensityExcitedStates}. Now, since $\varepsilon_3$ and $\mu$ are independent of the pair density $\rho> 0$ (which only affects the sequence $\mu_L$) we conclude, comparing~\eqref{EquationProofBECPairs} and~\eqref{EquationThermodynamicLimit}, that 
\begin{equation}
\limsup_{L \rightarrow \infty}\frac{1}{e^{\beta\left(E_L(0)-\mu_L\right)}-1} > 0 
\end{equation}
for a large enough pair density $\rho > 0$, therefore proving the statement.
\end{proof}
\begin{remark} Note that Theorem~\ref{CondensationPairs} can be reformulated in terms of a critical temperature rather than a critical density, see for example \cite{RuelleSM,BolteKernerBEC}.
\end{remark}
\subsection{On the effect of disorder on the condensate of pairs}
In this final subsection we want to investigate the effect of disorder in the metal on the condensate of pairs which exists in the free gas according to Theorem~\ref{CondensationPairs}. 

From a physical perspective it is reasonable to assume that local impurities in the metal lead to a modified interacting between the electron forming a bound pair (Cooper pair). In particular, if the strength of the  electron-phonon-electron interaction is decreased, the Coulomb repulsion leads to an effective repulsion between the two electrons. To model such an additional (spatially localised) two-particle interaction we follow \cite{KM16,KM16Rep}. More explicitly, the pair Hamiltonian $H^L_{p,int}$ with interaction shall be given by
\begin{equation}\label{HamiltonianCooperPairSingularInteraction}
H^L_{p,int}=-\frac{\hbar^2}{2m_e}\left(\frac{\partial^2}{\partial x^2}+\frac{\partial^2}{\partial y^2}\right)+v_{bin}(|x-y|)+v(x,y)[\delta(x)+\delta(y)]\ ,
\end{equation}
with $v:\mathbb{R}^2_+ \rightarrow \mathbb{R}_+$ being a symmetric potential and $\delta(\cdot)$ the Dirac-Delta distribution. Since the ``support'' of the $\delta(\cdot)$ is concentrated on $x=0$, the (repulsive) two-particle interactions are localised at the origin on the half-line. Of course, one could also consider shifted potentials with $\delta(\cdot-a)$, $a > 0$, but since the effect of such interactions is really important only for localised pair states we restrict attention to the case where $a=0$.

The quadratic form associated with $H^L_{p,int}$ is given by 
\begin{equation}\label{QFInteraction}
q^{int}_L[\varphi]=\frac{\hbar^2}{2m_e}\int_{\Omega_L}|\nabla \varphi|^2\ \ud x + \int_{\partial \Omega_{\sigma}} \sigma(y) |\varphi|_{\partial \Omega_{\sigma}}|^2\ \ud y
\end{equation}
with domain $\cD_{q_L}=\{\varphi \in H^1(\Omega_L)\ |\  \varphi(x,y)=-\varphi(y,x)\ \text{and}\ \varphi|_{\partial \Omega^L_D}=0  \}$ and $\sigma(y):=v(0,y)$.

According to [Theorem~1,\cite{KM16Rep}], the quadratic form \eqref{QFInteraction} is well-defined for potentials $\sigma \in L^{\infty}(0,d)$. Since $\Omega_L$ is bounded, the spectrum of the operator $H^L_{p,int}$ is purely discrete. Again, its eigenvalues shall be denoted by $E_L(n)$, $n \in \mathbb{N}_0$, with and the corresponding eigenstates by $\varphi_L(n)$.

The important effect of the singular interactions then is as follows: Given $\|\sigma\|_{L^{\infty}(0,d)}$ is large enough, the discrete part of the spectrum of $H^{L=\infty}_{p,int}$ (i.e. the corresponding Hamiltonian on the full domain $\Omega$ without anti-symmetry) becomes trivial [Theorem~4,\cite{KM16Rep}]. In analogy to this result we establish the following statement.
\begin{lemma}\label{LemmaConstant} There exists a constant $\gamma > 0$ such that $H^L_{p,int}$ has no eigenvalue smaller than $\frac{\hbar^2 \pi^2}{m_ed^2}$ whenever $\gamma \leq \|\sigma\|_{L^{\infty}(0,d)}$.
\end{lemma}
\begin{proof} The proof follows the lines of the proof of [Theorem~4,\cite{KM16Rep}].
\end{proof}
On the other hand, since the additional interaction term in \eqref{HamiltonianCooperPairSingularInteraction} affects only the lower triangular part of the domain $\Omega_L$, one readily obtains (again using a bracketing argument similar to the one used in the proof of Theorem~\ref{ExistenceOneEigenvalue}) the same formula \eqref{ParticleDensityExcitedStates} for the pair density in the excited states. However, since there exists no eigenvalue smaller than $\frac{\hbar^2 \pi^2}{m_ed^2}$ if $\|\sigma\|_{L^{\infty}(0,d)}$ is large enough, we can prove the following result.
\begin{theorem}[Absence of pair condensation]\label{AbsenceCondensation} Assume that $\sigma \in L^{\infty}(0,d)$ is such that $\gamma \leq \|\sigma\|_{L^{\infty}(0,d)}$ with $\gamma > 0$ as described in Lemma~\ref{LemmaConstant}. Then, for any inverse temperature $\beta > 0$ and all pair densities $\rho> 0$ there exists no condensation of pairs in the thermodynamic limit.
\end{theorem}
\begin{proof} We prove the statement by contradiction and assume therefore that the ground state $\varphi_L(0)$ of the Hamiltonian $H^L_{p,int}$ associated to the eigenvalue $E_L(0) \geq \frac{\hbar^2 \pi^2}{m_ed^2}$ is macroscopically occupied in the thermodynamic limit (note here that the ground state is always occupied the most).
	
	Now, this implies the existence of a subsequence of chemical potentials, which we shall also denote by $\mu_L$ for simplicity, for which $\mu_L \rightarrow \frac{\hbar^2 \pi^2}{m_ed^2}$ as $L \rightarrow \infty$ since $E_L(0) \rightarrow \frac{\hbar^2 \pi^2}{m_ed^2}$ as $L \rightarrow \infty$. Note that this can be proved similar to Lemma~\ref{LemmaConvergenceGroundstate}. However, for an arbitrary and sufficiently small $\varepsilon > 0$ we obtain 
	\begin{equation}\label{EquationProofInteractionCondensation}
	\frac{1}{L}\sum_{n=1}^{\infty}\frac{1}{e^{\beta\left(E_L(n)-\mu_L\right)}} \geq \rho^{ex}_{\infty}\left(\beta,\mu=\frac{\hbar^2 \pi^2}{m_ed^2}-\varepsilon\right) -\varepsilon_4(L)
	\end{equation}
	with $\varepsilon_4(L) \rightarrow 0$ as $L \rightarrow \infty$ by eq.~\eqref{ParticleDensityExcitedStates}.
	
	Now, due to eq.~\eqref{EquationThermodynamicLimit}, the left-hand side of eq.~\eqref{EquationProofInteractionCondensation} is bounded by the pair density $\rho > 0$. However, since 
	\begin{equation}
	\lim_{\varepsilon \rightarrow 0}\rho^{ex}_{\infty}\left(\beta,\mu=\frac{\hbar^2 \pi^2}{m_ed^2}-\varepsilon\right)=\infty
	\end{equation}
	by eq.~\eqref{ParticleDensityExcitedStates}, we can make the right-hand side of eq.~\eqref{EquationProofInteractionCondensation} arbitrarily large by first choosing $\varepsilon$ small and $L$ subsequently large enough. This, however, is a contradiction hence proving the statement. 
\end{proof}

\vspace*{0.5cm}


\vspace*{0.5cm}

{\small
\bibliographystyle{amsalpha}
\bibliography{Literature}}

\end{document}